\documentclass[runningheads]{llncs}

\usepackage[comma,authoryear,round,longnamesfirst,sectionbib]{natbib}

\usepackage{graphicx}
\usepackage{wrapfig}
\usepackage{booktabs}
\usepackage{amsmath,amsfonts,amssymb}
\usepackage{stmaryrd}%\part{title}
\usepackage[ruled,linesnumbered, noend]{algorithm2e}
\usepackage{turnstile}
\usepackage{mathrsfs}
\usepackage{enumitem}
\usepackage{amsthm}
\usepackage[super]{nth}
\usepackage[toc,page]{appendix}
\usepackage{listings}
\usepackage{mathtools}
\usepackage{csquotes}
\usepackage{multicol}
\usepackage{fancyvrb}
\usepackage{makecell}
\usepackage{mdframed}
\usepackage[toc,page]{appendix}
\usepackage{mathdots}
\usepackage{empheq}
\usepackage{appendix}
\usepackage{bm}
\usepackage{pgfplots}
\usepackage{enumitem}
\usepackage{float}

\usepackage{refcount}

\usepackage{array}
\newcolumntype{L}[1]{>{\raggedright\let\newline\\\arraybackslash\hspace{0pt}}m{#1}}
\newcolumntype{C}[1]{>{\centering\let\newline\\\arraybackslash\hspace{0pt}}m{#1}}
\newcolumntype{R}[1]{>{\raggedleft\let\newline\\\arraybackslash\hspace{0pt}}m{#1}}

\usepackage{tikz}
\usetikzlibrary{automata, positioning, calc, shapes, arrows, fit}

\usepackage[unicode=true, bookmarks=false, breaklinks=true, pdfborder={0 0 1}, colorlinks=true, citecolor=blue]{hyperref}

\usepackage{pifont}
\definecolor{MyGreen}{rgb}{0, 0.7, 0}
\definecolor{MyRed}{rgb}{0.8, 0, 0}

\usepackage{caption}
\usepackage{boxedminipage}
\usepackage{xspace}

\AtEndDocument{\refstepcounter{proposition}\label{finalprop}}

\title{Participatory Funding Coordination:\\ Model, Axioms and Rules}

\author{Haris Aziz\inst{1,}\inst{2}  \and Aditya Ganguly\inst{1}}

\institute{UNSW Sydney, Australia  \\
\email{\{haris.aziz,a.ganguly\}@unsw.edu.au}
\and
Data61 CSIRO\\
\medskip
}

\makeatletter
\newcommand{\@chapapp}{\relax}%
\makeatother

\begin{document}
	\maketitle

	\begin{abstract}
We present a new model of collective decision making
that captures important crowd-funding and donor coordination scenarios. In the setting, there is a set of projects (each with its own cost) and a set of agents (that have their budgets as well as preferences over the projects). An outcome is a set of projects that are funded along with the specific contributions made by the agents. For the model, we identify meaningful axioms that capture concerns including fairness, efficiency, and participation incentives. We then propose desirable rules for the model and study, which sets of axioms can be satisfied simultaneously. An experimental study indicates the relative performance of different rules as well as the price of enforcing fairness axioms.
	\end{abstract}
	
%------------------------------------------------------------------------
	
		\section{Introduction}
		Consider a scenario in which a group of house-mates want to pitch in money to buy some common items for the house but not every item is of interest or use to everyone. 
		Each of the items (e.g. TV, video game console, music system, etc.) has its price.  
		Each resident would like to have as many items purchased that are useful to her. She may have concerns about whether she is getting enough value for the contribution she makes. It is a scenario that is encountered regularly in numerous shared houses or apartments. 

		As a second scenario, hundreds of donors want to fund charitable projects.
		Each of the projects (e.g. building a well, enabling a surgery, funding a scholarship, etc.)
		has a cost requirement. Donors care about coordinating their donations in a way to fund commonly useful projects and they care about the amount of money that is used towards projects that they approve. How to coordinate the funding in a principled and effective way is a fundamental problem in crowdfunding and donor coordination. The model that we propose is especially suitable for coordinating donoations from alumni at our university. 

		Both of the settings above are coordination problems in which agents contribute money and they have preferences over the social outcomes. A collective outcome specifies which projects are funded and how much agents are charged. 

		\paragraph{Contributions.}
		We propose a formal model that we refer to as \emph{Participatory Funding Coordination (PFC)} that captures many important donor coordination scenarios. In this model, agents have an upper budget limit. They want as many of their approved projects funded.
We lay the groundwork for work on the model by formulating new axioms for the model. The logical relations between the axioms are established and the following question is studied:  which sets of axioms are simultaneously achievable. We propose and study rules for the problems that are inspired by welfarist concerns but satisfy participation constraints. 
		In addition to an axiomatic study of the rules, we also undertake an experimental comparison of the rules. The experiment sheds light on the impact that various fairness or participation constraints can have on the social welfare. This impact has been referred to as the price of fairness in other contexts. In particular, we investigate the effects of enforcing fairness properties on instances that model real-world applications of PFC, including crowdfunding.

	%------------------------------------------------------------------------

		\section{Related Work} 
		Our model generally falls under the umbrella of a collective decision making setting in which agents' donations and preferences are aggregated to make funding decisions. It is a concrete model within the broad agenda of achieving effective altruism~\citep{Maca15a,Maca17a,Peter19a}.

		The model we propose is  related to the discrete participatory budgeting model~\citep{AzSh19,ALT18a,GKS+19a,fain2016core,TaFa19a}. In discrete participatory budgeting, agents do not make personal donations towards the projects. They only express preferences over which projects should be funded.
		We present several axioms that are only meaningful for our model and not for discrete participatory budgeting. Algorithms for discrete participatory budgeting cannot directly be applied to our setting because they do not take into account individual rationality type requirements. 
	
			Another related setting is multi-winner voting~\citep{EFSS17a}. Multi-winner voting can be viewed as a restricted version of discrete participatory budgeting. The Participatory Funding Coordination (PFC) setting differs from multi-winner voting in some key respects: in our model, each project (winner) has an associated cost, and we select projects subject to a knapsack constraint as opposed to having a fixed number of winners. 

		Our PFC model relies on approval ballots in order to elicit agents' preferences. Dichotomous preferences have been considered in several important setting including committee voting~\citep{LaSk19a,ABC+16a} and discrete participatory budgeting~\citep{ALT18a,FSTW19}.

		Another related model that takes into account the contributions of agents was studied by \citet{BBPSS20a}. Just like in our model, an agent's utilities are based on how much money is spent on projects approved by the agent. However, their model 
		does not have any costs and agents can spread their money over projects in any way. %Hence, the projects can be viewed as infinitely divisible with unlimited cap on how much money that can be invested in them. 
	% Their model better fits coordination on long term charitable projects. However, in many crowd funding scenarios `indivisible' projects are more well-motivated, especially when the projects require a target funding for the success of the project (building a well, funding a scholarship etc.). 
	Our model has significant differences from the model of \citep{BBPSS19,BBPSS20a}: (1) in our setting, the projects are indivisible and have a minimum cost to complete, (2) agents may not be charged the full amount of their budgets. The combination of these features leads to challenges in even defining simple individual rationality requirements. Furthermore, it creates difficulties in finding polynomial-time algorithms for some natural aggregation rules (utilitarian, egalitarian, Nash product, etc.). 
		Our model is more appropriate for coordinating donations where projects have short-term deadlines and a target level of funding which must be reached for the project to be successfully completed. We show that the same welfarist rules that satisfy some desirable properties in the model \citep{BBPSS19,BBPSS20a}, fail to do so in our model. Just as the work of \citet{BBPSS19,BBPSS20a}, \citet{BHW19a} consider donor coordination for the divisible model in which the projects do not have costs and agents do not have budget limits. They also assume quasi-linear utilities whereas we model charitable donors who are not interested in profit but want their money being used as effectively as possible towards causes that matter to them. 
	 % as opposed to the continuous model which is framed in the setting of donations to charities that can accept arbitrarily large donations and scale their operations accordingly.% There is a case to be made then that the coordination of donations is even more important in the discrete setting, as short-term projects that miss out on their funding are deemed to have little or no value, whereas for long-term projects, one might expect that in the long run they will eventually raise the funds they require.

		The features of our PFC model enable the model to translate smoothly to a number of natural settings. Crowdfunding, in particular, is a scenario in which we would like to capitalise upon commonalities in donors' charitable preferences \citep{CCV15}. Furthermore, crowdfunding projects (e.g. building a well, funding a scholarship, etc.) often have provision points (see e.g. \citep{ACG13a,CGN16,DMCG19}), and it can be critical for these targets to be met (for example, a project to raise funds for a crowdfunding recipient to pay for a medical procedure would have to raise a minimum amount of money to be successful, otherwise all donations are effectively wasted). 

		Crowdfunding projects have been discussed in a broader context with various economic factors and incentive issues presented~\citep{ACG13a}. 
		%also suffer from various economic factors and incentives presented by \citet{ACG13a} that are linked to uncertainty over provision points and a lack of coordination between donors. 
		\citet{BaLi89a} discuss additional fairness and economic considerations for the related topic of the division of public goods. The discrete model that we explore, where projects have finite caps, has potential to coordinate donors and increase the effectiveness of a crowdfunding system. 

		%\citet{DMCG19}
		%\citet{BaLi89a}
		%\citet{ACG13a}

	%------------------------------------------------------------------------

		\section{Participatory Funding Coordination}
		A \emph{Participatory Funding Coordination (PFC)} setting is a tuple $(N,C,A,b,w)$ where $N$ is the set of agents/voters, $C$ is the set of projects (also generally referred to as candidates).
		The function $w: C \rightarrow \mathbb{R^+}$ specifies the cost $w(c)$ of each project $c\in C$. 
		The function $b: N \rightarrow \mathbb{R^+}$ specifies the budget $b_i$ of each agent $i\in C$. The budget $b_i$ can be viewed as the maximum amount of money that agent $i$ is willing to spend. For any set of agents $S\subseteq N$, we will denote $\sum_{i\in S}b_i$ by $b(S)$.  The approval profile $A=(A_1,\ldots, A_n)$ specifies for each agent, her set of acceptable projects $A_i$. 
	An \emph{outcome} is a pair $(S,x)$ where $S\subseteq C$ is the set of funded projects and $x$ is a vector of payments that specify for each $i\in N$, the payment $x_i$ that is charged from agent $i$. We will restrict our attention to feasible outcomes in which $x_i\leq b_i$ for all $i\in N$ and only those projects get financial contributions that receive their required amount. Also, note that the projects that are funded are only those that receive the entirety of their price in payments from the agents. 
		For any given PFC instance, a mechanism $F$ returns an outcome. We will denote the set of projects selected by $F$ as $F_C$ and the payments by $F_x$.\footnote{PFC can also be viewed as a matching problem in which the money of agents is matched to projects.}
		For any outcome $(S,x)$, since $x_i\leq b_i$, the money $b_i-x_i$ can either be kept by the agent $i$ or it can be viewed as going into some common pool. The main focus of our problem is to fund a maximal set of projects while satisfying participation constraints. 

		We suppose that an agent's preferences are \emph{approval-based}. 
		For any set of funded projects $S$, any agent $i$'s utility is $$u_i(S)=\sum_{c\in  S\cap A_i}w(c).$$ That is, an agent cares about how many dollars are \emph{usefully} used on his/her approved projects. Our preferences domain is similar to the one used by \citet{BBPSS20a} who considered a continuous model in which projects do not have target costs. In their model, agents also care about how much money is used for their liked projects.

		\begin{example}
		The following is an instance of a PFC problem with 5 agents and 6 projects. The costs of the projects is stated next to the project name. The budget of each agent is mentioned in front of the agent name. The plus sign indicates the approval of an agent for a project. 
	
			\begin{table}[h]
			\centering
					\scalebox{0.85}{
			\begin{tabular}{lcccccccc} \hline
			            & & A (7) & B (6) & C (1) & D (1) & E (8) & F (7) \\ \hline
						 & Budget	&  &  &  & &  & \\ \hline
			    Agent 1 & 3 	    	& + &   & + &   & + &    \\ \hline
			    Agent 2 & 3	   	& + &   &   & + & + &    \\ \hline
			    Agent 3 & 3     	&   & + & + &   &   & +  \\ \hline
			    Agent 4 & 2     	&   & + &   & + &   & +  \\ \hline
			    Agent 5 & 1     	& + &   &   &   & + &    \\ \hline
			\end{tabular}
			}
			\caption{Example of an PFC instance. }
			\label{table: EGAL-IMP PO-Pay Profile}
			\end{table}
			\end{example}

		%An agent may also care about the extent to which his payments are used for projects that are not approved by him or her. We will capture these ideas in some of the axiomatic properties we formalize. 

	%------------------------------------------------------------------------

		\section{Axiom design}

		In this section, we design axioms for outcomes of the PFC setting. We consider an outcome $(S,x)$. For any axiom $\mathbf{Ax}$ for outcomes, we say that a mechanism satisfies $\mathbf{Ax}$ if it always returns an outcome that satisfies $\mathbf{Ax}$.

		We first present three axioms for our setting that are based on the principle of participation:

		\begin{itemize}
			\item \textbf{Minimal Return (MR)}: each agent's utility is at least much as the money put in by the agent: $u_i(S)\geq x_i$. In other words, the societal decision is as good for each agent $i$ as $i$'s best use of the money $x_i$ that she is asked to contribute. We will use this as a minimal condition for all feasible outcomes.\footnote{One can strengthen MR to a stronger version in which $u_i(S,x)> x_i$ for each $i\in N$.}
		%	\item \textbf{Strict Minimal Return (SMR)}: each agent's utility is strictly more than the money put in by the agent: $u_i(S)> x_i$. In other words, the societal decision is better for each agent $i$ as $i$'s best use of the money $x_i$. 
			%	\item \textbf{Implementability (IMP)}: there exists a vector of payments $y$ of money such that each project receives a payment from an agent who approves it and $y_i=x_i$ for all $i\in N$.
			\item \textbf{Implementability (IMP)}  : There exists a payment function $y:N\times C\rightarrow \mathbb{R}^+ \cup \{0\}$ such that $\sum_{c\in C}y(i,c)=x_i$ for all $i\in N$ and $\sum_{i\in N}y(i,c) \in \{ 0, w(c) \}$ and there exists no $i\in N$ and $c\notin A_i$ such that $y(i,c)>0$. IMP captures the requirement that an agent's contribution should only be used on projects that are approved by the agent. 
			\item \textbf{Individual Rationality (IR):} the utility of an agent is at least as much as an agent can get by funding alone:
				$u_i(S)\geq \max_{S'\subseteq A_i, w(S')\leq b_i}(w(S')).$
				Note that IR is easily achieved if the project costs are high enough: if for $i\in N$ and $c\in C$, $w(c)>b_i$, then every outcome is IR.
		\end{itemize}		
		
		We note that MR is specified with respect to the amount $x_i$ charged to the agent. It can be viewed as a participation property: an agent would only want to participate in the market if she gets at least as much utility as the money she spent. We will show IMP is stronger than MR. IMP can also be viewed as a fairness property: agents are made to coordinate but they only spend their money on the projects they like. 
		
		% \item \textbf{Strong Individual Rationality (SIR)}: the utility of an agent is at least as much as agent can get by funding alone while also benefitting from the societal outcome: $u_i(F(N, b, C, A,b,w))\geq \max_{S'\subset A_i, w(S')\leq b_i}(w(S'))+ u_i(F(N\setminus \{i\}, b, C,, A,b,w))$.
		%\item \textbf{Individually Exhaustive (I-EXH)}: no agent can put in more money and fund another liked project: there exists no $c\in C\setminus S$ and $i\in N$ such that such that $c\in A_i$ and $w(c)\leq b_i-x_i$.

		% \begin{remark}
		% It is worth noting that in the continuous setting of funding coordination (Brandl et al., 2019), the axioms of MR and IMP are equivalent. In our discrete setting, the difference between these axioms comes from the fact that there are caps on project costs which prevent agents from putting arbitrarily large amounts of money into their preferred projects.
		% \end{remark}

		\begin{remark}
		If there is an IMP outcome where a set of projects are funded, then there is also an IMP outcome where any subset of these projects are funded.
		In order to find an IMP outcome for any subset, simply take the original outcome and set the payments of agents to projects that are being ``de-funded'' to zero.
		\end{remark}
	
		Next, we present axioms that are based on the idea of efficiency. 

		\begin{itemize}
			\item \textbf{Exhaustive (EXH)}: There exists no set $N'\subseteq N$ and project $c\in C\setminus S$ such that $c\in \cap_{i\in N'}A_i  \cap (C \setminus A)$ such that $w(c)\leq \sum_{i\in N'}(b_i-x_i)$.
		In words, agents in $N'$ cannot pool in their unspent money and fund another project liked by all of them. 
		% 	\item \textbf{Rationally Exhaustive (REx)}: {There exists no project $c\in C\setminus S$ such that at least one agent approves of $c$ and $w(c)\leq \sum_{i\in N}(b_i-x_i)$.
		% In words, agents cannot pool in their unspent money and fund another project that someone approves of.}
			\item \textbf{Pareto optimality (PO)-X}: An outcome $(S,x)$ is Pareto optimal within the set of outcomes satisfying property X if there exists no $(S',x')$ satisfying X such that $u_i(S')\geq u_i(S)$ for all $i\in N$ and  $u_i(S')> u_i(S)$ for some $i\in N$. Note that Pareto optimality is a property of the set of funded projects $S$ irrespective of the payments.
		% There exists no other outcome which Pareto dominates it. There exists no $(S',x')$ such that $u_i(S',x')\geq u_i(S,x)$ for all $i\in N$ and  $u_i(S',\textbf{x}')> u_i(S,x)$ for some $i\in N$.
			\begin{itemize}
				\item PO is Pareto optimal among the set of all outcomes. 
				\item PO-IMP: PO among the set of IMP outcomes.
				\item PO-MR: PO among the set of MR outcomes.
				% \item PO-MRIR: PO among the set of outcomes that are both MR and IR.
				% \item PO-IMPIR: PO among the set of outcomes that are both IMP and IR.
			\end{itemize}
			\item \textbf{Payment constrained Pareto optimality (PO-Pay)}: An outcome is PO-Pay if it is not Pareto dominated by any outcome of at most the same price. Formally, there exists no $(S',x')$ such that $\sum_{i\in N}x_i'\leq \sum_{i\in N}x_i$, $u_i(S')\geq u_i(S)$ for all $i\in N$ and  $u_i(S')> u_i(S)$ for some $i\in N$.
			\item \textbf{Weak Payment constrained Pareto optimality (weak PO-Pay)}: An outcome is weakly PO-Pay if it is not Pareto dominated by any outcome of at most the same price. Formally, there exists no $(S',x')$ such that $x_i'\leq x_i$ and $u_i(S')\geq u_i(S)$ for all $i\in N$ and  $u_i(S')> u_i(S)$ for some $i\in N$. 
			% \item \textbf{Pareto optimality (PO)}: There exists no other outcome which Pareto dominates it. There exists no $(S',x')$ such that $u_i(S')\geq u_i(S,x)$ for all $i\in N$ and  $u_i(S',x')> u_i(S,x)$ for some $i\in N$.
			% \begin{itemize}
			% 	\item PO-IMP: PO among the set of IMP outcomes
			% 	\item PO-MR: PO among the set of MR outcomes
			% 	\item PO-MRIR: PO among the set of outcomes that are both MR and IR
			% 	\item PO-IMPIR: PO among the set of outcomes that are both IMP and IR
			% 	\item PO: PO among the set of all outcomes
			% \end{itemize}
			% \item \textbf{Payment constrained Pareto optimality (PO-payment)}: %PO among the set of all outcomes in which the payment function $x$ does not increase pointwise. There exists no other outcome which Pareto dominates it. There exists no $(S',x)$ such that $u_i(S',x)\geq u_i(S,x)$ for all $i\in N$ and  $u_i(S',x)> u_i(S,x)$ for some $i\in N$. \textcolor{blue}{I think this axiom might need to be rethought. If the payment function is fixed, then this implies we are considering only sets of projects $S'$ that have the same total cost as $S$. Given that projects have different costs, it seems somewhat arbitrary to try and find other sets of projects that have exactly the same total cost as $S$, if such a set even exists.}
		\end{itemize}
	
		A concept that can be viewed in terms of participation, efficiency, and fairness is the adaptation of the principle of core stability for our setting. 
	
		\begin{itemize}
			\item \textbf{Core stability (CORE)}:
			There exists no set of agents who can pool in their budget and each gets a strictly better outcome. In other words, an outcome $(S)$ is CORE if for every subset of agents $N' \subseteq N$, for every subset of projects $C' \subseteq C$ such that $w(C') \leq \sum_{i \in N'} b_i$, the following holds for some agent $i \in N'$: $u_i(S)\geq w(C' \cap A_i).$
		\end{itemize}
	
	We describe a basic fairness axiom for outcomes and rules based on the idea of proportionality. 

		\begin{itemize}
			\item \textbf{Proportionality (PROP)}: Suppose a set of agents $N'\subseteq N$ \emph{only} approve of a set of projects $C'\subseteq C$ such that $\sum_{i\in N'}b_i\geq w(C')$. In that case, all the projects in $C'$ are selected.
		\end{itemize} 
				
	Finally, we consider an axiom that is defined for mechanisms rather than outcomes. We say that a mechanism satisfies \textbf{strategyproofness} if there exists no instance under which some agent has an incentive to misreport her preference relation.

		We conclude this section with some remarks on computation. The following proposition follows via a reduction from the Subset Sum problem. 

		\begin{proposition}\label{prop:comp}
		Even for one agent, computing an IR, PO, PO-MR, or PO-IMP outcome if NP-hard. 
		\end{proposition}
		\begin{proof}
		Consider the Subset Sum problem in which there is a set of items $M=\{1,\ldots, m\}$ with  corresponding weights $w_1,\ldots,w_m$, and a real value $W$. The problem is to find a subset $S$ with maximum weight $\sum_{j\in S}w_j$ such that $\sum_{j\in S}w_j\leq W$. The problem is well-known to be NP-hard. We reduce it our setting for a single agent by taking an item for each project that our agent approves of, and choosing the item weights to be the corresponding project costs. Then any set of projects $S$ satisfies the axioms in the proposition if and only if the corresponding set of items is the solution to the Subset Sum problem.
		\end{proof}
	
		Note that IMP is a property of an outcome not a set of projects. We say that a set of projects $S$ is IMP if there exists a feasible vector of charges to agents $x$ such that the outcome $(S,x)$ is IMP. The propoerty IMP can be tested in polynomial time via reduction to network flows. 
	 
		\begin{proposition}\label{prop:testimp}
		For a given set of projects $S$, checking whether there exists a vector of charges $x$ such that $(S,x)$ is implementable can be done in polynomial time.
		\end{proposition}

		% \begin{proof}
% 		In order to check whether a given set of projects $W$ is implementable, we just need to check whether the following linear program has a feasible solution or not. Note that $y$ represents a virtual payment function from each agent to each project that is compatible with the final agent charges $x$. The following can also be checked via network flows.
% 		\begin{align*}
% 		&y_{i,j}=0 &\text{ for all }i\in N, j\in [m] \text{ s.t. } p_j\notin A_i \\
% 		&\sum_{i\in N}y_{i,j}=w(p_j)& \text{ for all } p_j\in W \\
% 		% &\sum_{i\in N}y_{i,j}=0& \text{ for all } p_j\notin W \\
% 		&\sum_{j\in C}y_{i,j}\leq b_i& \text{ for all } i\in N \\
% 		&y_{i,j}\geq 0& \text{ for all }i\in N, j\in [m]\\
% 		\end{align*}
% 		\end{proof}
		
		Similarly, we can also check whether a particular outcome $(S, x)$ is implementable with a variation of the above linear program, where the upper bound on the sum of payments for each agent is $x_i$ instead of $b_i$.

		\section{Axioms: Compatibility and Logical Relations}

		In this section, we study the compatibility and relations between the axioms formulated. 

		\begin{remark}
		Note that IR and MR are incomparable. Any outcome in which an agent does not pay any money trivially satisfies MR. However, it may not satisfy IR. 
		% Suppose an outcome $(S,x)$ is not MR. This means that
		% $u_i(S)< x_i$ for some agent $i$ where $x_i\leq b_i$.
		On the other hand, an IR outcome may not be MR. Consider the case in which an agent's utility is at least as high as by funding alone. However, the agent may have been asked to pay more than the utility she gets which violates MR.
		\end{remark}

		Next, we point out that that PO-Pay is equivalent to weak PO-Pay.

		\begin{proposition}\label{po-paytoweak}
		PO-Pay is equivalent to weak PO-Pay.
		\end{proposition}
		\begin{proof}
		Suppose an outcome $(S,x)$ is not weakly PO-Pay. Then, it is trivially not PO-Pay.
		Now suppose $(S,x)$ is not PO-Pay. Then, there exists another outcome $(S',x')$ such that $\sum_{i\in N}x_i'\leq \sum_{i\in N}x_i$, $u_i(S')\geq u_i(S)$ for all $i\in N$ and  $u_i(S')> u_i(S)$ for some $i\in N$. Note that  $S'$ can be funded with total amount $\sum_{i\in N}x_i'$ irrespective of who paid what. So $S'$ is still affordable if $x_i'\leq x_i$.
		\end{proof}
		The next proposition establishes further logical relations between the axioms.
	
		\begin{proposition}\label{Implications}
		The following logical relations hold between the properties. 
		\begin{enumerate}%[label=\normalfont(\roman*)]
			\item IMP implies MR.
			\label{IMP implies MR}
			\item PO implies PO-Pay.
			\item PO-$X$ implies PO-$Y$ if $Y$ implies $X$.
			\item PO-IMP implies EXH.
			\item PO-IR implies EXH.
			\item CORE implies IR.
			\item The combination of PO-IMP and IMP imply PROP.
		\end{enumerate}
		\end{proposition}

	Next, we show that MR is compatible with PO-Pay.

		\begin{proposition}\label{prop:mrpo}
		Suppose an outcome is MR and there is no other MR outcome that Pareto dominates it. Then, it is PO-Pay.
		\end{proposition}
		\begin{proof}
		Suppose the outcome $(S,x)$ is MR and PO constrained to MR. 
		% Then, we claim that there does not exist another outcome $(S',x')$ such that $S'$ Pareto dominates $S$ and $x_i'\leq x_i$. If such an outcome $(S',x')$ existed then $u_i(S')\geq u_i(S)\geq x_i\geq x_i'$. Hence $(S',x')$ also satisfies MR and $S$ Pareto dominates $S'$ which contradicts the fact that $(S,x)$ PO constrained to MR.
		We claim that $(S,x)$ is PO-Pay. Suppose it is not PO-Pay. Then there exists another outcome $(S',x')$ such that $\sum_{i\in N}x_i'\leq \sum_{i\in N}x_i$, $u_i(S')\geq u_i(S)$ for all $i\in N$ and  $u_i(S')> u_i(S)$ for some $i\in N$. Note that $S'$ is afforadable with total amount $\sum_{i\in N}x_i'$ irrespective of who paid what. So $S'$ is still affordable if $x_i'\leq x_i$. Therefore, we can assume that $x_i'\leq x_i$ for all $i\in N$. 
		% "  So $S'$ is still affordable if $x_i'\leq x_i$ " for all $i \in N$
		Note that since $S'$ Pareto dominates $S$ and since $(S,x)$ is MR, $u_i(S')\geq u_i(S)\geq x_i\geq x_i'$ for all $i\in N$. Hence $(S',x')$ also satisfies MR. Since $(S',x')$ is MR and since $S$ Pareto dominates $S'$, it contradicts the fact that $(S,x)$ PO constrained to MR.
		\end{proof}

		\begin{proposition}
		There always exists an outcome that satisfies IMP, IR, PO-IMP and hence also MR and EXH.
		\end{proposition}
		\begin{proof}
		\textit{Existence of an outcome that satisfies IMP, IR, PO-IMP}: For each $i\in N$ compute $(S_i,y_i)$ that is an IR outcome. 
		This can be computed by finding a maximum total weight set of projects that has weight at most $b_i$.
		Then consider the outcome $(\bigcup_{i\in N}S_i,(y_1,\ldots, y_n))$.
		In such an outcome, we also keep track of which agent contributed to which project. Note that if $c\in S_i$, then $i$ contributed $w(c)$ to that project. Note that $w(\bigcup_{i\in N}S_i)\geq \sum_{i\in N}y_i$. If $w(\bigcup_{i\in N}S_i)> \sum_{i\in N}y_i$, we need to return $w(\bigcup_{i\in N}S_i) - \sum_{i\in N}y_i$, back to the agents to ensure that no more money is charged than needed to pay for $\bigcup_{i\in N}S_i$. We return the money as follows. 
		Recall that we know the amount paid by each agent to each project, i.e., agent $i$ paid $w(c)$ to project $c$ if and only if $c\in S_i$.
		Some projects may have received more money than needed. For each project $c$'s surplus, we uniformly allocate it among the agents who paid for it.
		Suppose the outcome satisfying IMP, IR and EXH does not satisfy PO-IMP. Then there exists another outcome that satisfies IMP that Pareto dominates the outcome. Such a Pareto improvement still satisfies IR because the utility of each agent is at least as high.
		\end{proof}
	
		Note that PO-Pay and IMP are both satisfied by an empty outcome with zero charges. PO-IMP and IMP are easily satisfied by computing a PO outcome from the set of IMP outcomes. PO-Pay and PO-IMP are easily satisfied by computing a PO outcome which may not necessarily satisfy IMP. 
	
		\begin{proposition}
		There always exists an outcome that satisfies MR, IR, PO-MR and hence also EXH.
		\end{proposition}
		\begin{proof}
		\textit{Existence of an outcome that satisfies MR, IR, PO-MR}: From the proof of part \textit{(i)}, we know that an IMP and IR outcome always exists. Also, from Proposition~\ref{Implications} we know that every IMP outcome is MR, so there always exists an MR and IR outcome. 
		Now suppose the outcome satisfying MR and IR does not satisfy PO-MR. Then there exists another outcome satisfying MR that Pareto dominates the original outcome, which is still IR. There cannot exist an infinite number of Pareto improvements because the budgets of the agents are finite. Hence we can reach a PO-MR outcome that is also IR and MR.
		\end{proof}
		%------------------------------------------------------------------------

		We note that if no agent can individually fund a project, then every outcome is IR. In crowdfunding settings in which projects have large costs, the IR requirement is often easily satisfied.

		\section{Aggregation Rules}
	
		In this section, we take a direct welfarist view to formalize rules that maximize some notion of welfare. We consider three notions of welfare: utilitarian, egalitarian, and Nash welfare; and we define the following rules.  
	
		% We translate and discuss the feasibility of aggregation rules that arise from the continuous funding coordination setting \citep{BBPSS19}. We discuss alternative approaches that are more appropriate to the PFC setting.

		\begin{itemize}
		% \item Weighted UTIL: define the weighted utilitarian welfare derived from an output $(S, x)$ as $$ \sum_{i\in N} x_i u_i(S). $$ Then, UTIL returns an outcome that maximises the utilitarian welfare.
		\item UTIL: define the utilitarian welfare derived from an outcome $(S, x)$ as $ \sum_{i\in N} u_i(S).$ Then, UTIL returns an outcome that maximises the utilitarian welfare.
		\item EGAL: given some output $(S, x)$, write the sequence of agents' utilities from that outcome as $u(S) = (u_i(S))_{i\in N}$, where $u$ is sorted in non-decreasing order. Then, EGAL returns an outcome $(S, x)$ such that $u(S)$ is lexicographically maximal among the outcomes. % Give an example to explain?
		\item NASH: maximises the Nash welfare derived from an output $(S, x)$, i.e. $\prod_{i\in N} \left( u_i(S)  \right).$
		% \haris{Why didn't you consider a weighted version of EGAL?}
		% i.e. $$ \{(S, x) : u(S, x) \geq_L u(S', x') \textrm{ for all possible } (S', x') \}. $$
		% \item weighted NASH: maximises the weighted Nash welfare derived from an output $(S, x)$, i.e. $$ \prod_{i\in N} \left( u_i(S)  \right)^{x_i}. $$
		% We can also explore a weighted version of Nash welfare as follows: compute an output $(S, x)$, i.e. $$ \prod_{i\in N} \left( u_i(S)  \right)^{x_i}. $$
		\end{itemize}

		\begin{proposition}
		UTIL, EGAL, and NASH satisfy PO and hence PO-MR, PO-IMP, PO-Pay, and EXH. 
		\end{proposition}
	
		One notes that the rules UTIL, EGAL, and NASH  do not satisfy minimal guarantees such as MR. The reason is that an agent may donate her budget to a widely approved project even though she may not approve any of such projects.
		Given that the existing aggregation rules do not provide us with guarantees that the outcomes they produce will satisfy our axioms, we can instead define rules such that optimize social welfare within certain subsets of feasible outcomes. For a property X, we can define UTIL-X, EGAL-X, and NASH-X as rules that maximise the utilitarian, egalitarian and Nash welfare respectively among only those outcomes that satisfy property X. 
	Next, we analyse the properties satisfied by rules EGAL/UTIL/NASH constrained to the set of MR or IMP outcomes. In the continuous model introduced by \citet{BBPSS20a}, there is no need to consider the rule NASH-IMP, as the NASH rule in the case where projects can be funded to an arbitrary degree (given there is sufficient budget) already satisfies IMP.

	Before we study the axiomatic properties, we note that 
	most meaningful axioms and rules are NP-hard to achieve or compute. The following proposition follows from Proposition~\ref{prop:comp}.

		\begin{proposition}
		Even for one agent, computing an UTIL, UTIL-MR, UTIL-MR, EGAL, EGAL-MR, EGAL-IMP, NASH, NASH-MR, NASH-IMP outcome is NP-hard.
		\end{proposition}
		% \begin{proof}
	 	% Consider the Subset Sum problem in which there is a set of items $M=\{1,\ldots, m\}$ with  corresponding weights $w_1,\ldots,w_m$, and a real values $W$. The problem is a subset $S$ with maximum weight $\sum_{j\in S}w_j$ such that $\sum_{j\in S}w_j\leq W$. The problem is well-known to be NP-hard. We reduce it our setting by taking a project for each corresponding item and project weight for each item weight. Then any set of projects $S$ satisfies the axioms in the proposition if and only if the corresponding set of items $S$ is the solution to the Subset problem.
		% \end{proof}
	
	Similarly, the following  proposition follows from Proposition~\ref{prop:mrpo}.

		\begin{proposition}
		UTIL-MR, EGAL-MR, and NASH-MR satisfy PO-Pay.
		\end{proposition}
		% \begin{proof}
		%
		% \end{proof}
		
	From Proposition~\ref{prop:mrpo}, it follows that UTIL-MR, EGAL-MR, and NASH-MR satisfy PO-Pay. In contrast, we show that UTIL-IMP, EGAL-IMP, and NASH-IMP do not satisfy PO-Pay. In order to show this, we prove that it is possible in some instances for the set of jointly IMP and PO-IMP outcomes to be disjoint from the set of PO-Pay outcomes. 

		\begin{proposition}\label{prop:po-pay}
		UTIL-IMP, EGAL-IMP and NASH-IMP do not satisfy PO-Pay. In fact it is possible that no IMP and PO-IMP outcome satisfies PO-Pay.
		\end{proposition}
		\begin{proof}
		Consider the following instance in Table~\ref{table: NASH-IMP PO-Pay Profile}.
		\begin{table}[h]
		\centering
		\begin{tabular}{lcccccccc} \hline
		            & & A (7) & B (4) & C (3) & D (7) & E (7)	 \\ \hline
					& Budget	&  	&  	&  	& 	& 	\\ \hline
		    Agent 1 	& 4 	    	&  	& + 	& 	&   	& + 	\\ \hline
		    Agent 2 	& 1	   	& + 	&   	&   	&   	& + 	\\ \hline
		    Agent 3 	& 1     	& + 	&   	&  	&   	& + 	\\ \hline
		    Agent 4 	& 5     	& + 	&   	&   	& + 	&   	\\ \hline
		    Agent 5 	& 3     	&   	&   	& + 	& + 	&   	\\ \hline
		\end{tabular}
		\caption{Example instance for proof of UTIL/EGAL/NASH-IMP not satisfying PO-Pay.}
		\label{table: NASH-IMP PO-Pay Profile}
		\end{table}	

		\begin{claim}
		Observe that no implementable outcome can fund project $E$ since it is too expensive to be funded solely by its supporters. 
		\end{claim}
	
		\begin{claim}
		Any implementable outcome funds a subset of the following project sets: $\{A, B, C\}, \{B, D\}$. Note that for an implementable outcome, if $D$ is funded, then only $B$ can also be funded (there is not enough money for agents who approve of $A$ or $C$ to fund these projects after funding $D$).
		\end{claim}
		 \vspace{-2em}
		\begin{table}[h!]
		\centering
		\begin{tabular}{lccccccc} \hline
		           	& A,B,C \quad & B,D \quad & D,E (not IMP)	\\ \hline
		    Agent 1 	& 4     	& 4     	& 7     	\\ \hline
		    Agent 2	& 7     	& 0     	& 7     	\\ \hline
		    Agent 3 	& 7    	& 0     	& 7     	\\ \hline
		    Agent 4	& 7     	& 7    	& 7     	\\ \hline
		    Agent 5 	& 3     	& 7     	& 7     	\\ \hline
		\end{tabular}
		\caption{Utilities provided to each agent by outcomes that fund the project sets $\{A, B, C\}, \{B, D\}$ and $\{D, E\} $.}
		\label{table: IMP PO-Pay Utilities}
		\end{table}
		\vspace{-2em}
		Note that when we are looking for the optimal outcome under a certain rule, we can ignore those project sets that are subsets of other project sets. Then, from Table \ref{table: IMP PO-Pay Utilities}, we see that the unique UTIL-IMP, EGAL-IMP and NASH-IMP outcome is the outcome that funds $\{A, B, C\}$ where each agent sends all their money to the only project of those three that they approve of. But clearly, this outcome is Pareto dominated by any outcome that funds $\{D, E\}$, which has the same total cost. Thus we have shown that UTIL-IMP, EGAL-IMP and NASH-IMP do not satisfy PO-Pay in general. In fact the set of IMP and PO-IMP outcomes can be disjoint from the set of PO-Pay outcomes. 
		\end{proof}
	
		The most striking aspect of Proposition~\ref{prop:po-pay} is that in the continuous domain without project costs, NASH-IMP is equivalent to NASH and the rule satisfies Pareto optimality and hence PO-Pay. In our context, NASH-IMP fails to satisfy PO-Pay.

		We also considered the issue of strategyproofness and found examples that show that none of the UTIL/EGAL/NASH rules are strategyproof whether they are unconstrained or constrained to MR or IMP outcomes. In contrast, there are several natural rules such as UTIL that are strategyproof in the continuous setting as well in the multi-winner voting setting. 
		
		\begin{proposition}
		UTIL, UTIL-MR, UTIL-IMP, NASH, NASH-MR and NASH-IMP are not strategyproof.
		\end{proposition}
		\begin{proof}
		Consider the instance given in Table \ref{tab: UTIL not SP}. Note that we only include project $Z$ for the purpose of making the NASH welfare of outcomes non-zero. Now, observe that it is impossible to fund all three projects, so our possible candidate project sets to be funded by the above rules are those where two projects get funded.
	
			\begin{table}[h!]
			\centering
			\begin{tabular}{lccccc} \hline
			            	& 		& X (10) 	& Y (4) & Z (9)	\\ \hline
						& Budget	&  		& 		&		\\ \hline
			    Agent 1 	& 8 		&      	& +     	& +		\\ \hline
			    Agent 2 	& 1	 	&       	& +   	& +		\\ \hline
			    Agent 3 	& 10		& + 		& +		& +		\\ \hline
			\end{tabular}
			\caption{Example instance where rules are not strategyproof.}
			\label{tab: UTIL not SP}
			\end{table}
	\vspace{-2em}
			\begin{table}[h!]
			\centering
			\begin{tabular}{lccccc} \hline
			            	& $\{X, Y\}$	& $\{X, Z\}$	& $\{Y, Z\}$		\\ \hline
		Utilitarian Welfare 	& 22			& 37	    		& 39    			\\ \hline
		Nash Welfare    	& 224		& 1539	 	& 2197			\\ \hline
			\end{tabular}
			\caption{Utilitarian and Nash welfares of certain project sets to be funded.}
			\label{tab: Performance of subsets against rules - SP}
			\end{table}
	
			We check that there is an implementable outcome that funds $\{Y, Z\}$, and find that the outcome where Agents 1 and 2 pay for $Z$ and Agent 3 pays for $Y$ is implementable. Hence, $\{Y, Z\}$ is the result of UTIL, UTIL-MR, UTIL-IMP, NASH, NASH-MR, NASH-IMP. Note that the utility for Agent 3 is 13.
	
			Now, suppose Agent 3 were to misrepresent her preferences as in Table \ref{tab: UTIL not SP - misrep}. Again, according to this new (perceived) instance, it is impossible for all projects to be funded, so in Table \ref{tab: Performance of subsets against rules - misrep - SP} we check the welfares produced by funding any two of the projects.
	
			\begin{table}[h!]
			\centering
			\begin{tabular}{lccccc} \hline
			            	& 		& X (10) 	& Y (4) & Z (9)	\\ \hline
						& Budget	&  		& 		&		\\ \hline
			    Agent 1 	& 8 		&      	& +     	& +		\\ \hline
			    Agent 2 	& 1	 	&       	& +   	& +		\\ \hline
			    Agent 3 	& 10		& + 		& 		& +
			\end{tabular}
			\caption{Instance where Agent 3 is misrepresenting her preferences.}
			\label{tab: UTIL not SP - misrep}
			\end{table}
	
			\begin{table}[h!]
			\centering
			\begin{tabular}{lccccc} \hline
			            	& $\{X, Y\}$	& $\{X, Z\}$	& $\{Y, Z\}$		\\ \hline
			    UTIL 	& 18			& 37	    		& 35    			\\ \hline
			    NASH    	& 160		& 1539	 	& 1521			\\ \hline
			\end{tabular}
			\caption{Perceived welfares of certain project sets to be funded if Agent 3 misrepresents her preferences.}
			\label{tab: Performance of subsets against rules - misrep - SP}
			\end{table}

			Since $\{X, Z\}$ can be funded by an implementable outcome where Agents 1 and 2 paying for $Z$ and Agent 3 paying for $X$, $\{X, Z\}$ is the result of UTIL, UTIL-MR, UTIL-IMP, NASH, NASH-MR, NASH-IMP. With this outcome, Agent 3 sees her utility rise to 19.

			Then, by misrepresenting her preferences, Agent 3 can cause the choice of the aforementioned rules to change from funding $\{Y, Z\}$ to funding $\{X, Z\}$, hence increasing her own utility.	Therefore, UTIL, UTIL-MR, UTIL-IMP, NASH, NASH-MR and NASH-IMP are not strategyproof.
	
		\end{proof}
		
		Similarly, the following also holds.

		\begin{proposition}\label{prop:egalnotsp}
		EGAL, EGAL-MR and EGAL-IMP are not strategyproof.
		\end{proposition}

		Table~\ref{table:axioms} shows the axioms that are satisfied by restricting the aggregation rules to optimising within the space of MR or IMP outcomes. \\
		\begin{table*}
		\centering
		\label{tab:compare} 
		\scalebox{0.7}{
		\begin{tabular}{lcccccccccccccc}
		\toprule
		& UTIL-MR & EGAL-MR & NASH-MR & UTIL-IMP & EGAL-IMP & NASH-IMP \\ \midrule
		MR		& \checkmark & \checkmark & \checkmark & \checkmark & \checkmark & \checkmark \\
		IMP		&--	&--	&-- & \checkmark	& \checkmark & \checkmark \\
		PROP		&--	&--	&-- & \checkmark	& \checkmark & \checkmark \\
		IR 		&--	&--	&-- &--	&--	&--	\\
		PO 		&--	&--	&-- &--	&--	&-- \\
		PO-MR	& \checkmark & \checkmark & \checkmark &--	&--	&-- \\
		PO-IMP	& \checkmark & \checkmark & \checkmark &\checkmark & \checkmark	& \checkmark \\
		PO-Pay	& \checkmark &\checkmark&\checkmark&--&--&--& \\
		EXH 		& \checkmark &\checkmark	&\checkmark & \checkmark	& \checkmark & \checkmark \\
		CORE 	&--	&--	&--	&--	&--	&-- \\
		SP		&--	&--	&--	&--	&--	&-- \\
		\bottomrule
		\end{tabular}%
		}
		\caption{Properties satisfied by UTIL-MR, EGAL-MR, NASH-MR, UTIL-IMP, EGAL-IMP and NASH-IMP.}
		\label{table:axioms}
		\end{table*}
		%------------------------------------------------------------------------
	
		\section{Experiment}

		In addition to the axiomatic study of the welfare-based rules, we undertake a simulation-based experiment to gauge the performance of different rules with respect to utilitarian and egalitarian welfare. Our study shows the impact of fairness axioms such as MR and IMP on welfare. 

		We generate random samples of profiles in order to simulate two potential real-world applications of PFC.

		\begin{enumerate}
		\item Share-house setting: In this example, we can imagine a group of house-mates pooling their resources to fund communal items for their house. We operate under the following assumptions:
		\begin{itemize}
			\item Number of agents from 3-6: this represents a reasonable number of house-mates in a share-house.
			\item Number of projects from 5-12: projects may include buying items such as tables, chairs, sofas, televisions, lights, kitchen appliances, washing machines, dryers, etc.
			\item Agent budgets from 300-600 and project costs from 50-1000. We base these costs on typical rent and furniture costs in Australia as well as costs of the above items in first and second-hand retailers. We expect that each agent brings some money to the communal budget, and would spend around one or two weeks' worth of rent on one-time communal expenses.
		\end{itemize}
			
		\item Crowdfunding setting: In this example, we imagine a relatively small number of expensive projects to be funded, and a large number of philanthropic donors, and make the following assumptions.
		\begin{itemize}
			\item Number of agents from 20-50: A review of crowdfunding websites such as Kickstarter and GoFundMe shows that the most promoted projects are typically funded by thousands of donors, and smaller projects can attract tens of donors. For the purposes of our simulation, we use between 20-50 donors, which is still relatively large to the number of available projects. 
			\item Number of projects from 3-8: In crowdfunding, there are far more projects available than a donor actually sees. However, we can estimate that in a browsing session, a donor might view the top 3-8 promoted projects.
			\item Agent budgets from 0-400 and project costs from 1000-10000: Projects in real-life crowdfunding can have vastly varying costs. For our simulation, we want for the agents with all their money combined to be able to afford some, but not all of the available projects in order to create instances that are not trivially resolved by funding all or none of the projects.
		\end{itemize}
		\end{enumerate}

		Imposing MR on a rule seems to have a significant impact on both utilitarian and egalitarian welfare on average. Of course, since IMP implies MR, we expect that imposing IMP as a constraint will have an even greater cost on welfare, but from our experiment, this cost is a relatively small increase on the the cost of imposing MR. It is worth noting that in worst-case scenarios, it is always possible that there are no non-trivial outcomes that satisfy the constraints, and so there is a risk that a rule subject to a constraint could produce an outcome that gives all agents zero utility.
		
		When considering average performance, rules are more resilient to the imposition of fairness constraints for instances that simulate crowdfunding scenarios compared to share-house scenarios. When the number of agents is large and the number of projects is small, and project costs are large compared to agent budgets, it seems to be easier to achieve fairness properties.
	
		We typically expect the NASH rule to be a compromise between UTIL and EGAL. This manifests in the results, where the performance losses for NASH with respect to 	utilitarian welfare are considerably less than those for EGAL. Likewise, NASH loses considerably less with respect to egalitarian welfare than UTIL.

		\section{Conclusions}

		We proposed a concrete model for coordinating funding for projects. A formal approach is important to understand the fairness, participation, and efficiency requirements a system designer may pursue. We present a detailed taxonomy of such requirements and clarify their properties and relations. We also analyse natural welfarist rules 
	both axiomatically and experimentally. Our model is not just a rich setting to study collective decision making. We feel that the approaches considered in the paper go beyond academic study and can be incorporated in portals that aggregate funding for charitable projects. We envisage future work on online versions of the problem. 
	
		%In the paper, we considered several axioms such as IR, MR, and IMP that pertain to strategic issues such as participation incentives. 
	%UTIL-IMP is also not SP. The reason is that in the context of approval-based committee voting, UTIL-IMP implies proportionality as considered by ~\citep{Pete18a}. \citet{Pete18a} proved that in approval-based multi-winner voting, proportionality is which is incompatible with strategyproofness~\citep{Pete18a}.

	In practical applications of PFC, it is important to balance welfare demands with fairness conditions. Our experiment investigated the cost of fairness when imposing MR or IMP on UTIL, EGAL and NASH rules over instances that model crowdfunding and share-house scenarios. We find that imposing MR alone significantly reduces welfare on average, but imposing IMP as well produces a relatively small additional cost on welfare. The costs of imposing any fairness condition are much more pronounced on instances that model a share-house setting than a crowdfunding setting, suggesting that for a large number of agents and large project costs, fairness conditions are more easily met.

	\bibliographystyle{plainnat}
	%\ liography{../../pamas/abb,../../pamas/group,../../pamas/brandt,../../pamas/aziz}

	% \bibliography{../adtbib/abbshort,../adtbib/adt,../adtbib/aziz,../adtbib/bib}

	\newpage

	\appendix

	\appendix

		\section{Remaining Proofs}
	
		\textbf{Proof of Proposition~\ref{prop:comp}}
	
		% \begin{proposition}\label{prop:comp}
		% Even for one agent, computing an IR, PO, PO-MR, or PO-IMP outcome if NP-hard.
		% \end{proposition}
	
		\begin{proof}
		Consider the Subset Sum problem in which there is a set of items $M=\{1,\ldots, m\}$ with  corresponding weights $w_1,\ldots,w_m$, and a real values $W$. The problem is a subset $S$ with maximum weight $\sum_{j\in S}w_j$ such that $\sum_{j\in S}w_j\leq W$. The problem is well-known to be NP-hard. We reduce it our setting by taking a project for each corresponding item and project weight for each item weight. Then any set of projects $S$ satisfies the axioms in the proposition if and only if the corresponding set of items $S$ is the solution to the Subset problem. 
		\end{proof}

		% \begin{proposition}\label{prop:testimp}
		% For a given set of projects $(S,x)$, checking whether there exists feasible payment $x$ such that $(S,x)$ can be done in polynomial time.
		% \end{proposition}
	
		\noindent \textbf{Proof of Proposition~\ref{prop:testimp}}
		\begin{proof}
		In order to check whether a given set of projects $W$ is implementable, we just need to check whether the following linear program has a feasible solution or not. The following can also be checked via network flows. 
	
		\begin{align*}
		&x_{i,j}=0 &\text{ for all }i\in N, j\in [m] \text{ s.t. } p_j\notin A_i \\
		&\sum_{i\in N}x_{i,j}=w(p_j)& \text{ for all } p_j\in W \\
		&\sum_{i\in N}x_{i,j}=0& \text{ for all } p_j\notin W \\
		&\sum_{j\in C}x_{i,j}\leq b_i& \text{ for all } i\in N \\
		&x_{i,j}\geq 0& \text{ for all }i\in N, j\in [m]\\
		\end{align*}	
		\end{proof}
	
		\noindent \textbf{Proof of Proposition~\ref{po-paytoweak}}

				% \begin{proposition}\label{po-paytoweak}
			% 	PO-pay is equivalent to weak PO-pay.
			% 	\end{proposition}
				\begin{proof}
				Suppose an outcome $(S,x)$ is not weakly PO-pay. Then, it is trivially not PO-pay.
				Now suppose $(S,x)$ is not PO-pay. Then, there exists another outcome $(S',x')$ such that $\sum_{i\in N}x_i'\leq \sum_{i\in N}x_i$, $u_i(S')\geq u_i(S)$ for all $i\in N$ and  $u_i(S')> u_i(S)$ for some $i\in N$. Note that the $S'$ is afforadable with total amount $\sum_{i\in N}x_i'$ irrespective of who paid what. So $S'$ is still affordable if $x_i'\leq x_i$.
				\end{proof}

				% \begin{proposition}\label{Implications}
				% The following logical relations hold between the properties.
				% \begin{enumerate}%[label=\normalfont(\roman*)]
				% 	\item IMP implies MR.
				% 	\label{IMP implies MR}
				% 	\item PO implies PO-Pay.
				% 	\item PO-$X$ implies PO-$Y$ if $Y$ implies $X$.
				% 	\item PO-IMP implies EXH.
				% 	\item PO-IR implies EXH.
				% 	\item CORE implies IR.
				% 	\item The combination of PO-IMP and IMP imply PROP.
				% \end{enumerate}
				% \end{proposition}
			
				\noindent \textbf{Proof of Proposition~\ref{Implications}}

				\begin{proof}
					We distinguish between the cases.
				\begin{enumerate}
		\item \textit{IMP implies MR}: Suppose an outcome $(S, x)$ satisfies IMP. Then there exists a set of vectors $\{ \bm{y}_c \}_{c\in S}$ where $\bm{y}_c$ is a vector of payments from each agent to project $c$ such that $\sum_{c\in S} \bm{y}_c = x$ and $y_{c, i} = 0$ if $c \notin A_i$. Examining any row $i$ of the vector $x$, which denotes the money charged from agent $i$, we see that $x_i = \sum_{c\in S} \bm{y}_{c, i}$. But since $y_{c, i} = 0$ if $c \notin A_i$, we have $x_i = \sum_{c\in S \cap A_i} y_{ij} $. Now, $$u_i(S) = \sum_{c\in S \cap A_i} w(c) = \sum_{c\in S \cap A_i}  \sum_{j\in N} y_{c, j} \geq \sum_{c\in S \cap A_i} y_{c, i} = x_i. $$ Hence, $(S, x)$ satisfies MR.
					\item \textit{PO implies PO-Pay}: Suppose an outcome is PO. Then, it is not Pareto dominated by any other outcome. Hence, it is not Pareto dominated by any outcome of lesser total cost, and so it is PO-Pay.
					\item \textit{PO-$X$ implies PO-$Y$ if $Y$ implies $X$}: Suppose some condition $Y$ implies another condition $X$. Now, suppose some outcome $(S, x)$ is PO-$X$. Then, $(S, x)$ is not Pareto dominated by any outcome that satisfies $X$. Since $Y$ implies $X$, the set of all outcomes satisfying $Y$ is a subset of the set of all outcomes satisfying $Y$. Thus, $(S, x)$ is not Pareto dominated by any outcome that satisfies $Y$, and so $(S, x)$ is PO-$Y$.
					\item \textit{PO-IMP implies EXH}: Suppose for a contradiction that $(S, x)$ is an outcome that satisfies PO-IMP but not EXH. Let an implementable payment function for this outcome be $y: N \times C \rightarrow \mathbb{R}^+ \cup \{0\}$. Since this outcome is not exhaustive, there is a set of agents $N'$ who can pool together their unspent money to fund another commonly-liked project $c'$. We can construct a new outcome $(S \cup \{c'\}, x')$ with a payment function $y': N \times C \rightarrow \mathbb{R}^+ \cup \{0\}$ such that for all agents $i \in N$ and projects $c \in S$, $y'(i, c) = y(i, c)$, and for all agents $j \in N'$, $y'(j, c') = \delta_j$, where $\delta_j$ is the contribution of each agent $j \in N'$ to the new project $c'$. Note that $(S \cup \{c'\}, x')$ is implementable since its payment function has agents only funding projects they approve of. Now, $(S \cup \{c'\}, x')$ Pareto dominates $(S, x)$ which is a contradiction since $(S, x)$ is PO-IMP. Therefore, any PO-IMP outcome must satisfy EXH.
					\item \textit{PO-IR implies EXH}: Suppose for a contradiction that $(S, \textbf{x})$ is an outcome that satisfies PO-IR but not EXH. Then there is a set of agents $N'$ who can pool together their unspent money to fund another commonly-liked project $c'$. Since no agent's utility decreases by funding this project, a new outcome $(S \cup \{ c'\}, \textbf{x}')$ is still IR, where $\textbf{x}'$ is any valid vector of payments. Also note that this $(S \cup \{ c'\}, \textbf{x}')$ Pareto dominates $(S, \textbf{x})$ since no agent's utility decreases and at least one agent's utility increases. This is a contradiction as $(S, \textbf{x})$ is PO-IR by our initial assumption. Hence, any PO-IR outcome is EXH.
					\item Suppose an outcome $(S, x)$ is CORE. Then for every subset of agents $N' \subseteq N$, for every subset of projects $C' \subseteq C$ such that $w(C') \leq \sum_{i \in N'} b_i$, for some agent $i \in N'$ $u_i(S)\geq w(C' \cap A_i).$ Now consider the case where $|N'| = 1$, i.e. $N'$ is a subset of one agent. We now have that for every agent $i$, for all $C' \subseteq C$ such that $w(C') \leq  b_i$, $u_i(S)\geq w(C' \cap A_i).$ Equivalently, for every agent $i$, $u_i(S)\geq \max_{S'\subseteq A_i, w(S')\leq b_i}(w(S'))$ and so $(S, x)$ is IR.
					\item The combination of PO-IMP and IMP imply PROP. Suppose an outcome does not satisfy PROP. Then this means that there is set of agents $N'\subseteq N$ \emph{only} approve of a set of projects $C'\subseteq C$ such that $\sum_{i\in N'}b_i\geq w(C')$ but not all projects in $C'$ are selected. Then one of the two cases occurs: (1) either the money of agents in $N'$ is used for projects not approved by them which violates IMP (2) the agents in $N'$ can pool in unspent money to fund an additional project in $C'$ that is not funded, which means that the outcome is not PO-IMP.
				\end{enumerate}
				\end{proof}
				
				\noindent \textbf{Proof of Proposition~\ref{prop:egalnotsp}}

						% \begin{proposition}\label{prop:egalnotsp}
	% 					EGAL, EGAL-MR and EGAL-IMP are not strategyproof.
	% 					\end{proposition}

						\begin{proof}
						Consider the instance given in Table \ref{tab: EGAL not SP}. Due the total budget constraint, at most two of the projects can be funded, so we check the egalitarian welfare derived by funding any two projects in Table \ref{tab: Performance of subsets against EGAL rules - SP}.

							\begin{table}[h!]
							\centering
							\begin{tabular}{lccccc} \hline
							            	& 		& X (3) 	& Y (2) & Z (1)	\\ \hline
										& Budget	&  		& 		&		\\ \hline
							    Agent 1 	& 1 		&      	& +     	& +		\\ \hline
							    Agent 2 	& 1	 	&       	& +   	& +		\\ \hline
							    Agent 3 	& 3		& + 		& +		& 		\\ \hline
							\end{tabular}
							\caption{Example instance where EGAL rules are not strategyproof.}
							\label{tab: EGAL not SP}
							\end{table}

							\begin{table}[h!]
							\centering
							\begin{tabular}{lccccc} \hline
							            		& $\{X, Y\}$	& $\{X, Z\}$	& $\{Y, Z\}$		\\ \hline
						Egalitarian Welfare 	& (2, 2, 5)	& (1, 1, 3)	& (2, 3, 3)		\\ \hline
							\end{tabular}
							\caption{Egalitarian welfares of certain project sets to be funded.}
							\label{tab: Performance of subsets against EGAL rules - SP}
							\end{table}

							Observe that it is possible for an implementable outcome to fund $\{Y, Z\}$ by having Agents 1 and 2 pay for them, and so $\{Y, Z\}$ is funded by each of the above rules. Then, the utility for Agent 3 is 2.

							Now, suppose Agent 3 misrepresents her preferences to suppress the fact that she approves of project $Y$. The new perceived instance is shown in Table \ref{tab: EGAL not SP - misrep} and again, we compute the egalitarian welfare produced by funding any two projects in Table \ref{tab: Performance of subsets against EGAL rules - misrep - SP}.

							\begin{table}[h!]
							\centering
							\begin{tabular}{lccccc} \hline
							            	& 		& X (3) 	& Y (2) & Z (1)	\\ \hline
										& Budget	&  		& 		&		\\ \hline
							    Agent 1 	& 1 		&      	& +     	& +		\\ \hline
							    Agent 2 	& 1	 	&       	& +   	& +		\\ \hline
							    Agent 3 	& 3		& + 		& 		& 		\\ \hline
							\end{tabular}
							\caption{Example instance where rules are not strategyproof.}
							\label{tab: EGAL not SP - misrep}
							\end{table}

							\begin{table}[h!]
							\centering
							\begin{tabular}{lccccc} \hline
							            		& $\{X, Y\}$	& $\{X, Z\}$	& $\{Y, Z\}$		\\ \hline
						Egalitarian Welfare 	& (2, 2, 3)	& (1, 1, 3)	& (0, 3, 3)		\\ \hline
							\end{tabular}
							\caption{Egalitarian welfares of certain project sets to be funded.}
							\label{tab: Performance of subsets against EGAL rules - misrep - SP}
							\end{table}

							Note that there is an implementable outcome that funds $\{X, Y\}$, where Agent 3 pays for $X$ and Agents 1 and 2 pay for $Y$. Hence, $\{X, Y\}$ is funded by each of the rules. The new utility for Agent 3 is 3.

							Thus, by misrepresenting her preferences, Agent 3 is able to increase the utility she receives when egalitarian rules are used. Therefore, EGAL, EGAL-MR and EGAL-IMP are not strategyproof.

							\end{proof}
				
				\section{Experiments}
				
				The results of the experiments are depicted in Figures~\ref{fig: UTIL_Average_Sharehouse}, \ref{fig: UTIL_Average_Crowdfunding}, 
				\ref{fig: UTIL_Worst_Sharehouse},
				\ref{fig: UTIL_Worst_Crowfunding},
			 	\ref{fig: EGAL_Average_Sharehouse},
			 	\ref{fig: EGAL_Average_Crowdfunding}, \ref{fig: EGAL_Worst_Sharehouse}, and 	\ref{fig: EGAL_Worst_Crowfunding}.
				
				\begin{figure*}
				    \makebox[\textwidth][c]{%
				    \begin{minipage}{1.2\textwidth}
				    \centering
				    \begin{minipage}{0.4\textwidth}
				        \centering
				        \includegraphics[width=\textwidth]{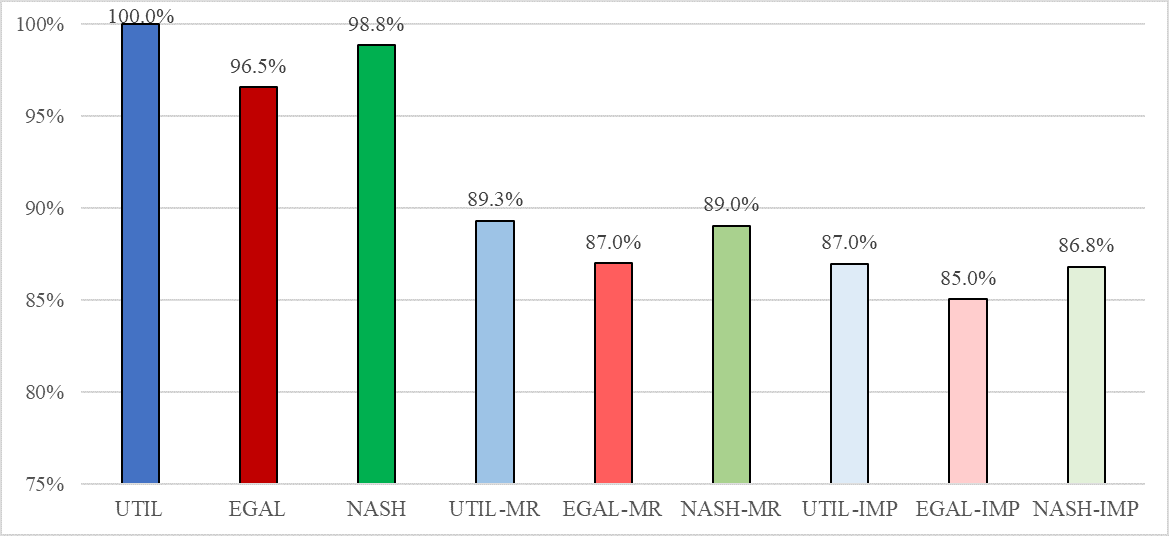}
						\caption{Average performance of rules with respect to utilitarian welfare in share-house simulations as a percentage of the maximum achievable utilitarian welfare.}
						\label{fig: UTIL_Average_Sharehouse}
				    \end{minipage}%
				    \hspace{0.03\textwidth}
				    \begin{minipage}{0.4\textwidth}
				        \centering
				        \includegraphics[width=\textwidth]{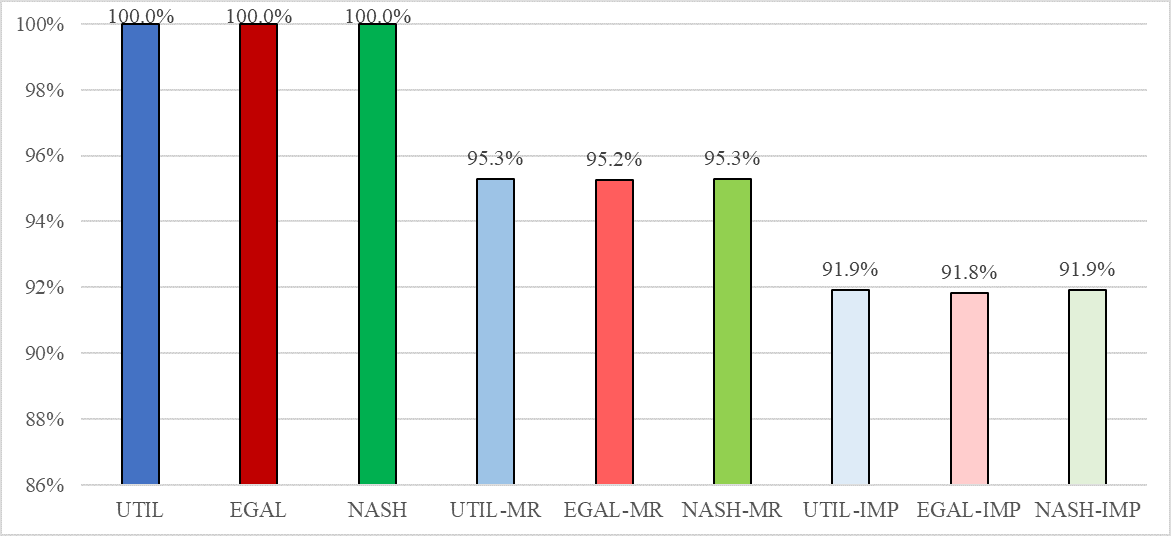}
						\caption{Average performance of rules with respect to utilitarian welfare in crowdfunding simulations as a percentage of the maximum achievable utilitarian welfare.}
						\label{fig: UTIL_Average_Crowdfunding}
				    \end{minipage}
				    \end{minipage}
				    }

					\vspace{0.02\textwidth}

					\makebox[\textwidth][c]{%
					\begin{minipage}{1.2\textwidth}
					\centering
				    \begin{minipage}{.4\textwidth}
				        \centering
				        \includegraphics[width=\textwidth]{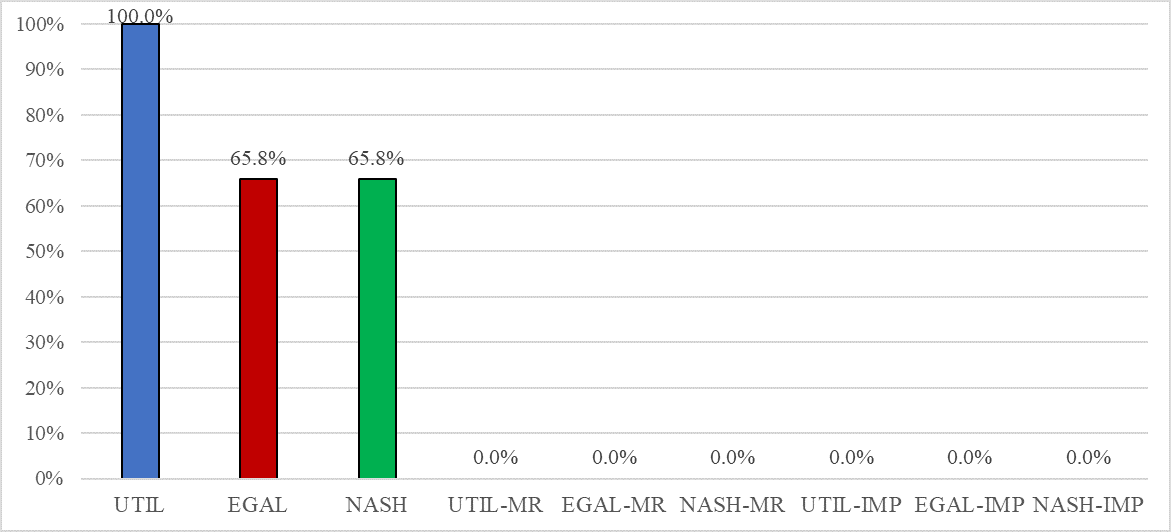}
						\caption{Worst-case performance of rules with respect to utilitarian welfare in share-house simulations as a percentage of the maximum achievable utilitarian welfare.}
						\label{fig: UTIL_Worst_Sharehouse}
				    \end{minipage}%
				    \hspace{0.03\textwidth}
				    \begin{minipage}{0.4\textwidth}
				        \centering
				        \includegraphics[width=\textwidth]{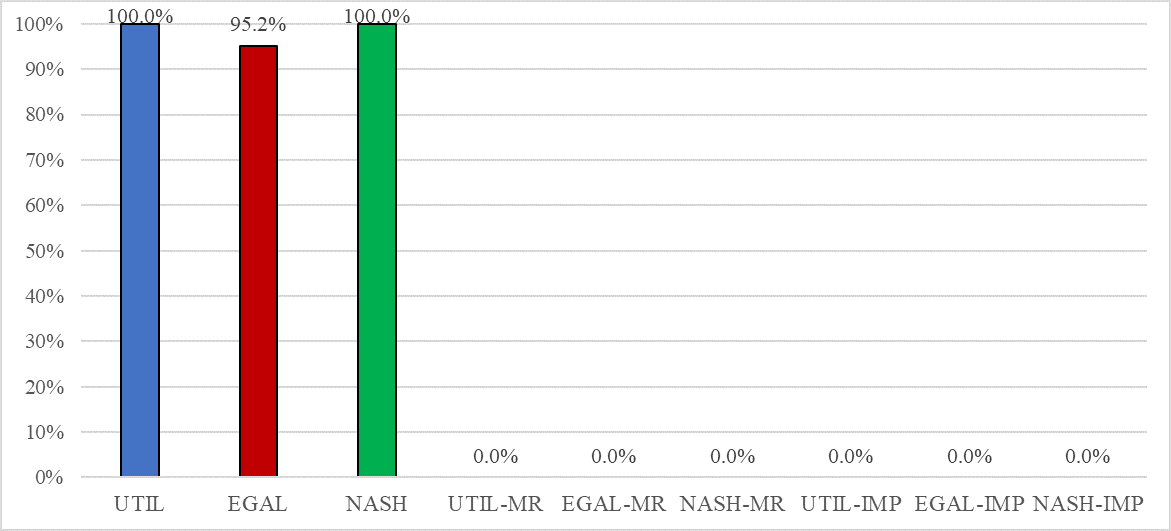}
						\caption{Worst-case performance of rules with respect to utilitarian welfare in crowdfunding simulations as a percentage of the maximum achievable utilitarian welfare.}
						\label{fig: UTIL_Worst_Crowfunding}
				    \end{minipage}
				    \end{minipage}
					}

					\vspace{0.02\textwidth}

					\makebox[\textwidth][c]{%
					\begin{minipage}{1.2\textwidth}
					\centering
				    \begin{minipage}{.4\textwidth}
				        \centering
				        \includegraphics[width=\textwidth]{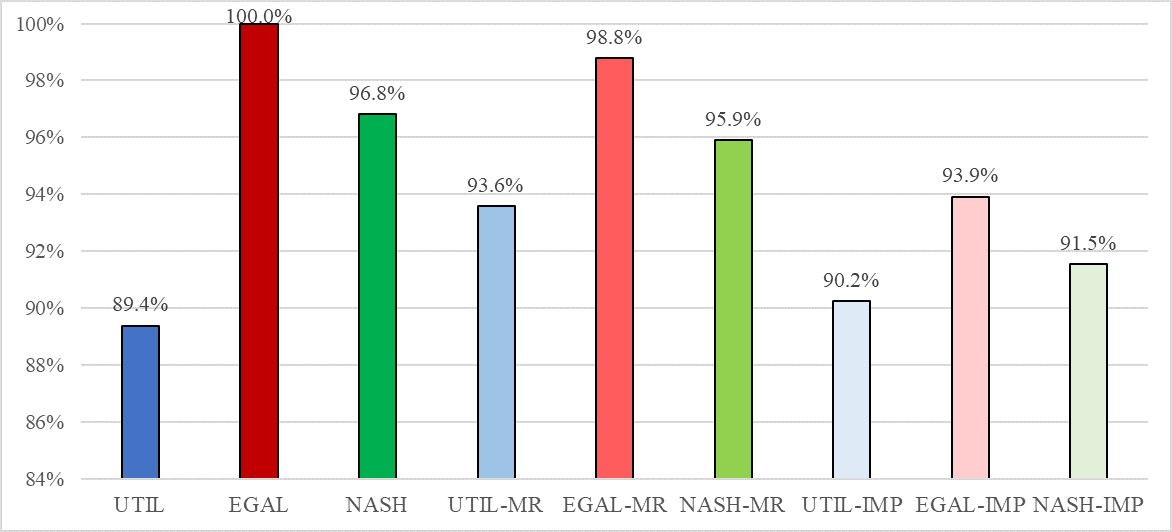}
						\caption{Average performance of rules with respect to egalitarian welfare in share-house simulations as a percentage of the maximum achievable egalitarian welfare.}
						\label{fig: EGAL_Average_Sharehouse}
				    \end{minipage}%
				    \hspace{0.03\textwidth}
				    \begin{minipage}{0.4\textwidth}
				        \centering
				        \includegraphics[width=\textwidth]{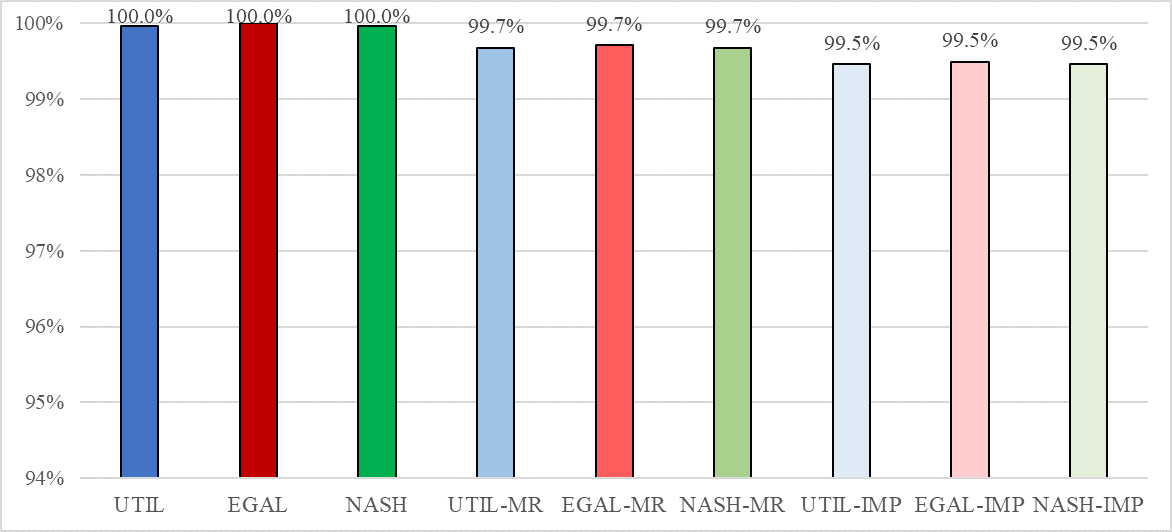}
						\caption{Average performance of rules with respect to egalitarian welfare in crowdfunding simulations as a percentage of the maximum achievable egalitarian welfare.}
						\label{fig: EGAL_Average_Crowdfunding}
				    \end{minipage}
				    \end{minipage}
					}

					\vspace{0.02\textwidth}

					\makebox[\textwidth][c]{%
					\begin{minipage}{1.2\textwidth}
					\centering
				    \begin{minipage}{.4\textwidth}
				        \centering
				        \includegraphics[width=\textwidth]{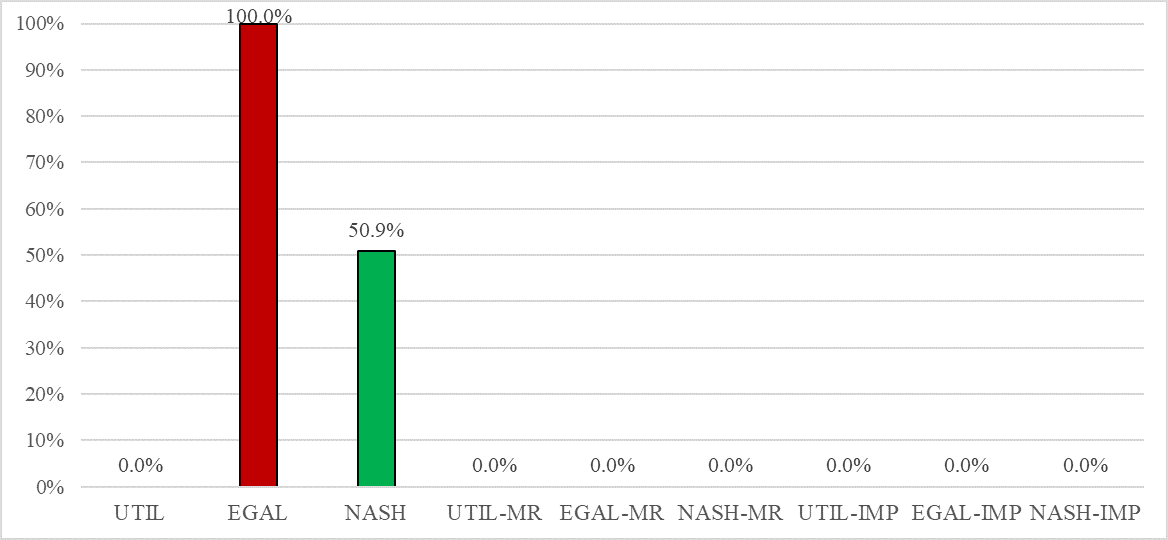}
						\caption{Worst-case performance of rules with respect to egalitarian welfare in share-house simulations as a percentage of the maximum achievable egalitarian welfare.}
						\label{fig: EGAL_Worst_Sharehouse}
				    \end{minipage}%
				    \hspace{0.03\textwidth}
				    \begin{minipage}{0.4\textwidth}
				        \centering
				        \includegraphics[width=\textwidth]{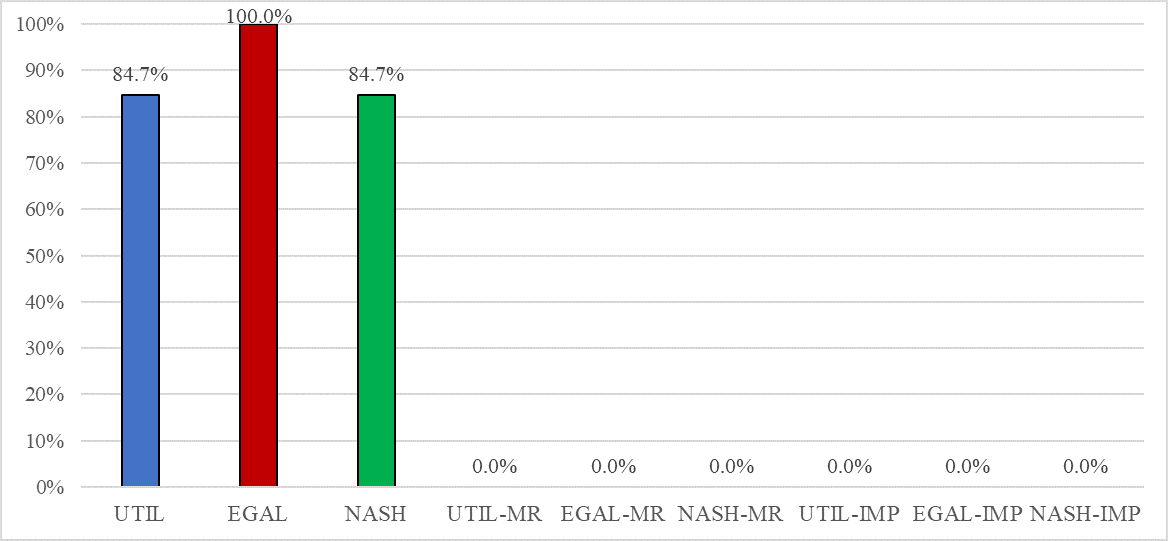}
						\caption{Worst-case performance of rules with respect to egalitarian welfare in crowdfunding simulations as a percentage of the maximum achievable egalitarian welfare.}
						\label{fig: EGAL_Worst_Crowfunding}
				    \end{minipage}
				    \end{minipage}
				    }

				\end{figure*}
	
		\section{Additional Propositions}

	% \setcounter{proposition}{\getrefnumber{finalprop}}
	% \addtocounter{proposition}{-1}
	
	\setcounter{proposition}{13}

		\begin{proposition}
		UTIL, EGAL and NASH satisfy PO.
		\end{proposition}
		\begin{proof}
		Suppose there was an outcome that was Pareto dominant over the outcome returned by any of these rules. Then, it would also have a strictly greater utilitarian/egalitarian/Nash welfare to this outcome, which is a contradiction.
		\end{proof}
	
		%\begin{proposition}
		%UTIL, EGAL and NASH do not satisfy MR (or IMP by corollary) or IR (or CORE by corollary).
		%\end{proposition}
		%\begin{proof}	
		%Consider the instance in Table \ref{tab: UTIL, EGAL, NASH not MR or IR}. Observe that given the total budget constraint, if $W$ is funded, no other projects can be funded, and if $W$ is not funded, then all the other projects can be funded. We then see that a UTIL outcome would fund $\{W\}$, whereas an EGAL or NASH outcome would fund $\{X, Y, Z\}$. In either case, Agent 3 gets a utility of at most 3, but would clearly be charged more than that since the other agents together cannot afford any of their approved projects entirely. Hence the UTIL, EGAL and NASH outcomes are not MR. Also, in the case of UTIL, Agent 3 is forced to pay towards project $W$, but if she left the system, she could fund project $Z$ herself and get a greater utility. Alternatively, in the case of EGAL or NASH, 
		%, $Z$, making this outcome not IR.
		%\begin{table}[h!]
		%\centering
		%\begin{tabular}{lccccc} \hline
		%            & & W (20) & X (10) & Y(3) & Z(3) \\ \hline
	%				&Budget	&  & \\ \hline
		%    Agent 1 & 1	&	& +	&	& +     \\ \hline
		%    Agent 2 & 1	& +	&  	& +	&		\\ \hline
		%    	Agent 3 & 20	&	&  	& 	& +		\\ \hline
		%\label{tab: UTIL, EGAL, NASH not MR or IR}    
		%\end{tabular}
		%\caption{Example instance where UTIL, EGAL and NASH outcomes do not satisfy MR or IR.}	
		%\end{table}
		%\end{proof}
	
		\begin{proposition}
		UTIL and NASH do not satisfy MR (or IMP by corollary) or IR (or CORE by corollary).
		\end{proposition}
		\begin{proof}
		For the profile in Table \ref{tab: UTIL, NASH not MR}, UTIL and NASH would require that only project $X$ is funded by using all of both agents' money. But then, Agent 2 could have left the mechanism and derived better utility on her own (IR). Also, her return is less than her contribution, so UTIL does not satisfy MR.
		\begin{table}[h!]
		\centering
		\begin{tabular}{lccccc} \hline
		            & & X (20) & Y (10) \\ \hline
					&Budget	&  & \\ \hline
		    Agent 1 & 10 	& +      &        \\ \hline
		    Agent 2 & 10 	&        & +  	  \\ \hline
		\end{tabular}
		\caption{Example profile where UTIL and NASH outcomes do not satisfy MR.}
		\label{tab: UTIL, NASH not MR}
		\end{table}
		\end{proof}
	
		\begin{proposition}
		EGAL does not satisfy MR (or IMP by corollary).
		\end{proposition}
		\begin{proof}
		For the profile in Table \ref{tab: UTIL, NASH not MR}, the maximally egalitarian output is to implement both projects. However, in this case, Agent 1 is charged $25$ but receives a utility of 20, so this outcome does not satisfy MR.
		\begin{table}[h!]
		\centering
		\begin{tabular}{lccccc} \hline
		            & & X (20) & Y (10) \\ \hline
					&Budget	&  & \\ \hline
		    Agent 1 & 25 	& +      &        \\ \hline
		    Agent 2 & 5	 	&        & +  	  \\ \hline
		\end{tabular}
		\caption{Example profile where EGAL outcomes do not satisfy MR.}
		\label{tab: UTIL, NASH not MR}
		\end{table}
		\end{proof}
	
		\begin{proposition}
		EGAL does not satisfy IR.
		\end{proposition}
		\begin{proof}
		In the below example, given the total budget, the choice is to implement either $X$ or $Y$ or neither. The egalitarian outcome with be to implement $Y$, but then, by leaving the system, Agent 1 could pay for $X$ herself and get a better outcome.  \\
		\begin{table}[h!]
		\centering
		\begin{tabular}{lccccc} \hline
		            & & X (20) & Y (10) \\ \hline
					&Budget	&  & \\ \hline
		    Agent 1 & 20 	& +      & +      \\ \hline
		    Agent 2 & 5	 	&        & +  	  \\ \hline
		\end{tabular}
		\end{table}
		\end{proof}
	
		\begin{proposition}
		UTIL-IMP, EGAL-IMP and NASH-IMP do not satisfy PO-MR.
		\end{proposition}
		\begin{proof}
		In Table \ref{tab: X-IMP not PO-MR}, the only implementable outcomes are those in which no projects are funded or Agents 1 and 2 pay for project $X$. But the outcome where all three projects are funded by all agents spending all of their money satisfies MR and Pareto dominates any of the implementable outcomes.
		\begin{table}[h!]
		\centering
		\begin{tabular}{lccccc} \hline
		            	& 		& X (10)	& Y (12)	\\ \hline
					& Budget	& 		&		\\ \hline
		    Agent 1 	& 10 	& +     	& 		\\ \hline
		    Agent 2 	& 10	 	& +    	& 		\\ \hline
		    Agent 3 	& 2		& 		& +		\\ \hline
		\end{tabular}
		\caption{Example profile where UTIL-IMP, EGAL-IMP and NASH-IMP outcomes do not satisfy PO-MR.}
		\label{tab: X-IMP not PO-MR}
		\end{table}
		\end{proof}
	
		\begin{proposition}
		UTIL-MR and NASH-MR do not satisfy IMP.
		\end{proposition}
		\begin{proof}
		 Consider the example below. UTIL-MR and NASH-MR will require that both projects are funded. However, at least one of Agents 1 and 2 must give some money to project $Y$, as Agent 3 cannot afford it by herself. Thus, this outcome is not implementable.
		\begin{table}[h!]
		\centering
		\begin{tabular}{lccccc} \hline
		            & & X (30) & Y (10)	\\ \hline
					&Budget	&  & \\ \hline
		    Agent 1 & 20 	& +      &  	      	\\ \hline
		    Agent 2 & 20	 	& +      &   	  	\\ \hline
		    Agent 3 & 5	 	&        & +  	  	\\ \hline
		\end{tabular}
		\end{table}
		\end{proof}
	
		\begin{proposition}
		UTIL-MR and UTIL-IMP do not satisfy IR (or CORE by corollary).
		\end{proposition}
		\begin{proof}
		Consider the example below. The overall (unique) utilitarian outcome is achieved by funding projects $Y$ and $Z$ with $x_1 = 4$ and $x_2 = 11$. Observe that this is an implementable outcome, and so this would be the result of UTIL-MR and UTIL-IMP. However, if Agent 1 left the system, she could have individually funded projects $W$ and $X$ which would have returned to her a greater utility. Therefore, this outcome does not satisfy IR.
		\begin{table}[h!]
		\centering
		\begin{tabular}{lccccc} \hline
		            & & W (3) & X (3) 	& Y (5) 	& Z (10)	\\ \hline
					&Budget	&  && \\ \hline
		    Agent 1 & 6 		& +     & +     	& + 		& 		\\ \hline
		    Agent 2 & 11	 	&       & 	   	& +		& +		\\ \hline
		\end{tabular}
		\end{table}
		\end{proof}
	
		\begin{proposition}
		NASH-MR, MASH-IMP, EGAL-MR and EGAL-IMP do not satisfy IR (or CORE by corollary).
		\end{proposition}
		\begin{proof}
		 From the example above, we see that any egalitarian distribution also will fund only projects $Y$ and $Z$. For this output to be implementable (and also satisfy MR), we have $x_1 = 4$ and $x_2 = 11$ or $x_1 = 5$ and $x_2 = 10$. In either case, we have seen from above that this outcome will not satisfy IR.
		 % Add the reasoning about NASH.
		\end{proof}

		\begin{itemize}

		\item NASH-MR and NASH-IMP do not satisfy IR (or CORE by corollary): Same argument as above.
	
		\end{itemize}
	
	\begin{proposition}
	UTIL, UTIL-MR, UTIL-IMP, NASH, NASH-MR and NASH-IMP are not strategyproof.
	\end{proposition}
	\begin{proof}
	Consider the instance given in Table \ref{tab: UTIL not SP}. Note that we only include project $Z$ for the purpose of making the NASH welfare of outcomes non-zero. Now, observe that it is impossible to fund all three projects, so our possible candidate project sets to be funded by the above rules are those where two projects get funded.
	
		\begin{table}[h!]
		\centering
		\begin{tabular}{lccccc} \hline
		            	& 		& X (10) 	& Y (4) & Z (9)	\\ \hline
					& Budget	&  		& 		&		\\ \hline
		    Agent 1 	& 8 		&      	& +     	& +		\\ \hline
		    Agent 2 	& 1	 	&       	& +   	& +		\\ \hline
		    Agent 3 	& 10		& + 		& +		& +		\\ \hline
		\end{tabular}
		\caption{Example instance where rules are not strategyproof.}
		\label{tab: UTIL not SP}
		\end{table}
	
		\begin{table}[h!]
		\centering
		\begin{tabular}{lccccc} \hline
		            	& $\{X, Y\}$	& $\{X, Z\}$	& $\{Y, Z\}$		\\ \hline
	Utilitarian Welfare 	& 22			& 37	    		& 39    			\\ \hline
	Nash Welfare    	& 224		& 1539	 	& 2197			\\ \hline
		\end{tabular}
		\caption{Utilitarian and Nash welfares of certain project sets to be funded.}
		\label{tab: Performance of subsets against rules - SP}
		\end{table}
	
		We check that there is an implementable outcome that funds $\{Y, Z\}$, and find that the outcome where Agents 1 and 2 pay for $Z$ and Agent 3 pays for $Y$ is implementable. Hence, $\{Y, Z\}$ is the result of UTIL, UTIL-MR, UTIL-IMP, NASH, NASH-MR, NASH-IMP. Note that the utility for Agent 3 is 13.
	
		Now, suppose Agent 3 were to misrepresent her preferences as in Table \ref{tab: UTIL not SP - misrep}. Again, according to this new (perceived) instance, it is impossible for all projects to be funded, so in Table \ref{tab: Performance of subsets against rules - misrep - SP} we check the welfares produced by funding any two of the projects.
	
		\begin{table}[h!]
		\centering
		\begin{tabular}{lccccc} \hline
		            	& 		& X (10) 	& Y (4) & Z (9)	\\ \hline
					& Budget	&  		& 		&		\\ \hline
		    Agent 1 	& 8 		&      	& +     	& +		\\ \hline
		    Agent 2 	& 1	 	&       	& +   	& +		\\ \hline
		    Agent 3 	& 10		& + 		& 		& +
		\end{tabular}
		\caption{Instance where Agent 3 is misrepresenting her preferences.}
		\label{tab: UTIL not SP - misrep}
		\end{table}
	
		\begin{table}[h!]
		\centering
		\begin{tabular}{lccccc} \hline
		            	& $\{X, Y\}$	& $\{X, Z\}$	& $\{Y, Z\}$		\\ \hline
		    UTIL 	& 18			& 37	    		& 35    			\\ \hline
		    NASH    	& 160		& 1539	 	& 1521			\\ \hline
		\end{tabular}
		\caption{Perceived welfares of certain project sets to be funded if Agent 3 misrepresents her preferences.}
		\label{tab: Performance of subsets against rules - misrep - SP}
		\end{table}

		Since $\{X, Z\}$ can be funded by an implementable outcome where Agents 1 and 2 paying for $Z$ and Agent 3 paying for $X$, $\{X, Z\}$ is the result of UTIL, UTIL-MR, UTIL-IMP, NASH, NASH-MR, NASH-IMP. With this outcome, Agent 3 sees her utility rise to 19.

		Then, by misrepresenting her preferences, Agent 3 can cause the choice of the aforementioned rules to change from funding $\{Y, Z\}$ to funding $\{X, Z\}$, hence increasing her own utility.	Therefore, UTIL, UTIL-MR, UTIL-IMP, NASH, NASH-MR and NASH-IMP are not strategyproof.
	
	\end{proof}

	\begin{proposition}
	EGAL, EGAL-MR and EGAL-IMP are not strategyproof.
	\end{proposition}
	\begin{proof}
	Consider the instance given in Table \ref{tab: EGAL not SP}. Due the total budget constraint, at most two of the projects can be funded, so we check the egalitarian welfare derived by funding any two projects in Table \ref{tab: Performance of subsets against EGAL rules - SP}.

		\begin{table}[h!]
		\centering
		\begin{tabular}{lccccc} \hline
		            	& 		& X (3) 	& Y (2) & Z (1)	\\ \hline
					& Budget	&  		& 		&		\\ \hline
		    Agent 1 	& 1 		&      	& +     	& +		\\ \hline
		    Agent 2 	& 1	 	&       	& +   	& +		\\ \hline
		    Agent 3 	& 3		& + 		& +		& 		\\ \hline
		\end{tabular}
		\caption{Example instance where EGAL rules are not strategyproof.}
		\label{tab: EGAL not SP}
		\end{table}
	
		\begin{table}[h!]
		\centering
		\begin{tabular}{lccccc} \hline
		            		& $\{X, Y\}$	& $\{X, Z\}$	& $\{Y, Z\}$		\\ \hline
	Egalitarian Welfare 	& (2, 2, 5)	& (1, 1, 3)	& (2, 3, 3)		\\ \hline
		\end{tabular}
		\caption{Egalitarian welfares of certain project sets to be funded.}
		\label{tab: Performance of subsets against EGAL rules - SP}
		\end{table}
	
		Observe that it is possible for an implementable outcome to fund $\{Y, Z\}$ by having Agents 1 and 2 pay for them, and so $\{Y, Z\}$ is funded by each of the above rules. Then, the utility for Agent 3 is 2.
	
		Now, suppose Agent 3 misrepresents her preferences to suppress the fact that she approves of project $Y$. The new perceived instance is shown in Table \ref{tab: EGAL not SP - misrep} and again, we compute the egalitarian welfare produced by funding any two projects in Table \ref{tab: Performance of subsets against EGAL rules - misrep - SP}.
	
		\begin{table}[h!]
		\centering
		\begin{tabular}{lccccc} \hline
		            	& 		& X (3) 	& Y (2) & Z (1)	\\ \hline
					& Budget	&  		& 		&		\\ \hline
		    Agent 1 	& 1 		&      	& +     	& +		\\ \hline
		    Agent 2 	& 1	 	&       	& +   	& +		\\ \hline
		    Agent 3 	& 3		& + 		& 		& 		\\ \hline
		\end{tabular}
		\caption{Example instance where rules are not strategyproof.}
		\label{tab: EGAL not SP - misrep}
		\end{table}
	
		\begin{table}[h!]
		\centering
		\begin{tabular}{lccccc} \hline
		            		& $\{X, Y\}$	& $\{X, Z\}$	& $\{Y, Z\}$		\\ \hline
	Egalitarian Welfare 	& (2, 2, 3)	& (1, 1, 3)	& (0, 3, 3)		\\ \hline
		\end{tabular}
		\caption{Egalitarian welfares of certain project sets to be funded.}
		\label{tab: Performance of subsets against EGAL rules - misrep - SP}
		\end{table}

		Note that there is an implementable outcome that funds $\{X, Y\}$, where Agent 3 pays for $X$ and Agents 1 and 2 pay for $Y$. Hence, $\{X, Y\}$ is funded by each of the rules. The new utility for Agent 3 is 3.
	
		Thus, by misrepresenting her preferences, Agent 3 is able to increase the utility she receives when egalitarian rules are used. Therefore, EGAL, EGAL-MR and EGAL-IMP are not strategyproof.
	
		\end{proof}
	
\end{document}